\documentclass[lettersize,journal]{IEEEtran}

\usepackage{amsmath,amsfonts}
\usepackage{amssymb}
\usepackage{graphicx}
\usepackage{cite}
\usepackage{bm}
\usepackage{tikz}
\usepackage{arydshln}
\usepackage{booktabs}
\usepackage{mathtools}
\usepackage{float}
\usetikzlibrary{arrows,positioning,calc}
\usetikzlibrary{arrows,arrows.meta,positioning,calc}

\usepackage{amsthm}
\newtheorem{theorem}{Theorem}
\newtheorem{lemma}{Lemma}

\newtheorem{definition}{Definition}
\newtheorem{assumption}{Assumption}
\newtheorem{remark}{Remark}

\newcommand{\ie}{\textit{i}.\textit{e}.}

\newcommand{\R}{\mathbb{R}}
\newcommand{\Rn}{\mathbb{R}^n}
\newcommand{\Rm}{\mathbb{R}^m}

\begin{document}

\title{Velocity-Free Horizontal Position Control of Quadrotor Aircraft via
Nonlinear Negative Imaginary Systems Theory}

\author{Ahmed G.~Ghallab and Ian R.~Petersen,~\IEEEmembership{Fellow,~IEEE}%
\thanks{This work was supported by the Australian Research Council under
grant DP230102443.}%
\thanks{A.~G.~Ghallab is with the Department of Mathematics, Faculty of Science, Fayoum University, Fayoum 63514, Egypt and I.~R.~Petersen is with the Research School of Engineering, The Australian National University, Canberra ACT 2601, Australia (e-mail: agg00@fayoum.edu.eg; i.r.petersen@gmail.com).}}

\markboth{IEEE Transactions on Control Systems Technology}%
{Ghallab \MakeLowercase{\textit{et al.}}: Velocity-Free Horizontal Position Control of Quadrotor Aircraft}

\maketitle

\begin{abstract}
This paper presents a velocity-free position control strategy for quadrotor unmanned aerial vehicles based on nonlinear negative imaginary (NNI) systems theory. Unlike conventional position control schemes that require velocity measurements or estimation, the proposed approach achieves asymptotic stability using only position feedback. We establish that the quadrotor horizontal position subsystem, when augmented with proportional feedback, exhibits the NNI property with respect to appropriately defined horizontal thrust inputs. A strictly negative imaginary integral resonant controller is then designed for the outer loop, and robust asymptotic stability is guaranteed through satisfaction of explicit sector-bound conditions relating controller and plant parameters. The theoretical framework accommodates model uncertainties and external disturbances while eliminating the need for velocity sensors. Simulation results validate the theoretical predictions and demonstrate effective position tracking performance.
\end{abstract}

\begin{IEEEkeywords}
Negative imaginary systems; quadrotor control; velocity-free control; robust stability; underactuated systems; nonlinear systems.
\end{IEEEkeywords}


\section{Introduction}

\IEEEPARstart{U}{nmanned} aerial vehicles (UAVs) have become increasingly important in both civilian and military applications, with quadrotors emerging as one of the most versatile platforms due to their mechanical simplicity, vertical takeoff and landing capabilities, and cost-effectiveness \cite{du2017distributed, ozbek2016feedback}. The control of quadrotor systems presents significant challenges due to their underactuated nature, nonlinear dynamics, and sensitivity to external disturbances. Consequently, researchers have developed various control strategies including backstepping \cite{bouadi2007modelling}, sliding mode control \cite{li2014adaptive}, distributed artificial neural networks \cite{tran2020distributed}, model predictive control \cite{alexis2012model}, and robust neural network-based approaches \cite{bouchoucha2008step}.

A critical aspect of autonomous quadrotor operation is accurate position control, which traditionally relies on both position and velocity feedback. However, obtaining reliable velocity measurements can be challenging in practice due to sensor noise, computational delays in numerical differentiation, and the additional cost and complexity of velocity sensors or observers. This motivates the development of velocity-free control approaches that achieve robust performance using only position feedback.

The negative imaginary (NI) systems framework, originally developed for flexible structure control with colocated force actuators and position sensors \cite{petersen2010}, offers a promising avenue for such designs. By exploiting the inherent passivity-like properties of mechanical systems, NI theory enables the design of controllers that guarantee robust stability without requiring velocity information. Recent advances have extended NI theory to nonlinear systems \cite{11060877}, opening new possibilities for application to nonlinear robotic systems. In our earlier work \cite{IP1}, we proposed a velocity-free attitude control scheme for quadrotors based on the nonlinear NI framework, wherein neither angular velocity measurements nor a model-based observer reconstructing the angular velocity is required. That contribution addresses the rotational subsystem responsible for attitude stabilisation; the present paper develops a complementary velocity-free strategy for the translational subsystem, targeting horizontal position control.

This paper extends the NNI framework to quadrotor horizontal position control, establishing theoretical foundations and demonstrating practical implementation through an inner-outer loop architecture. The main contributions are as follows. First, we prove that the quadrotor horizontal position subsystem with proportional feedback satisfies the NNI property. Second, we design a strictly negative imaginary controller that achieves asymptotic stability without velocity measurements. Third, we derive explicit sector-bound conditions relating controller parameters to system parameters. Finally, we validate the theoretical results through numerical simulations demonstrating robust performance.

The remainder of this paper is organized as follows. Section~\ref{sec:prelim} reviews the NI systems theory for both linear and nonlinear systems. Section~\ref{sec:quadrotor} presents the quadrotor dynamical model. Section~\ref{sec:control} develops the control design methodology and establishes the main stability result. Section~\ref{sec:simulation} presents simulation results, and Section~\ref{sec:conclusions} concludes the paper.

\section{Preliminaries on Negative Imaginary Systems}\label{sec:prelim}

Negative imaginary (NI) systems theory, introduced in \cite{petersen2010}, provides a robust control framework for systems with colocated force actuators and position sensors, particularly flexible structures. The theory has seen significant advances in both theoretical development and applications over the past decade \cite{petersen2016negative, lanzon2017feedback}. This section reviews relevant definitions and stability results from the NI literature for both linear and nonlinear systems.

\subsection{Linear Negative Imaginary Systems}

Consider the linear time-invariant (LTI) system
\begin{subequations}\label{eq:LTI_system}
\begin{align}
\dot{x}(t) &= Ax(t)+Bu(t), \label{eq:xdot1}\\
y(t) &= Cx(t)+Du(t), \label{eq:y1}
\end{align}
\end{subequations}
where $A \in \R^{n \times n}$, $B \in \R^{n \times m}$, $C\in \R^{m \times n}$, and $D\in \R^{m \times m}$. The system \eqref{eq:LTI_system} has the $m\times m$ real-rational proper transfer function matrix
\begin{equation}
G(s):=C(sI-A)^{-1}B+D.
\end{equation}

The frequency domain characterization of the NI property is given as follows.

\begin{definition}[\cite{lanzon2008}]\label{def:NI}
A square transfer function matrix $G(s)$ is called \emph{negative imaginary} (NI) if:
\begin{enumerate}
\item $G(s)$ has no pole at the origin and in $\Re[s]>0$;
\item For all $\omega >0$ such that $j\omega$ is not a pole of $G(s)$,
\begin{equation}
j\left( G(j\omega )-G^{*}(j\omega )\right) \geq 0,
\end{equation}
where $G^{*}(j\omega) := \overline{G^T(j\omega)}$ denotes the conjugate transpose;
\item If $j\omega_{0}$ with $\omega_0\in(0,\infty)$ is a pole of $G(s)$, it is at most a simple pole and the residue matrix
\begin{equation}
K_{0}:= \lim_{s\rightarrow j\omega_{0}}(s-j\omega_{0})sG(s)
\end{equation}
is positive semidefinite Hermitian.
\end{enumerate}
\end{definition}

A linear time-invariant system \eqref{eq:LTI_system} is NI if its transfer function is NI according to Definition~\ref{def:NI}. An equivalent time-domain characterization is provided in the following lemma.

\begin{lemma}[\cite{ghallab2018extending}]\label{lem:time_domain_NI}
Suppose that the system \eqref{eq:LTI_system} with $D=0$ is controllable and observable. Then $G(s)$ is negative imaginary if and only if there exists a positive definite matrix $P=P^T>0$ satisfying
\begin{equation}\label{eq:LMI_conditions}
PA+A^TP\leq 0, \quad PB=C^T
\end{equation}
such that the storage function $V(x)=\frac{1}{2}x^TPx$ satisfies
\begin{equation}\label{eq:dissipation_inequality}
    \dot{V}(x(t))\leq \dot{y}^T(t)u(t), \quad \forall \ t\geq0
\end{equation}
along system trajectories.
\end{lemma}

A strict version of the negative imaginary property is defined next.

\begin{definition}[\cite{lanzon2008}]\label{def:SNI}
A square transfer function matrix $G(s)$ is called \emph{strictly negative imaginary} (SNI) if:
\begin{enumerate}
\item $G(s)$ has no poles in $\Re[s]\geq0$;
\item $j[G(j\omega)-G^{*}(j\omega)]>0$ for all $\omega\in(0,\infty)$.
\end{enumerate}
\end{definition}

\subsection{Nonlinear Negative Imaginary Systems}

Consider the MIMO nonlinear system
\begin{subequations}\label{eq:nonlinear_system}
\begin{align}
\dot{\mathbf{x}} &= \mathbf{f}(\mathbf{x},\mathbf{u}),\label{eq:x}\\
\mathbf{y} &= \mathbf{h}(\mathbf{x}),\label{eq:y}
\end{align}
\end{subequations}
where $\mathbf{f}:\Rn\times \Rm\rightarrow\Rn$ is Lipschitz continuous with $\mathbf{f}(\mathbf{0},\mathbf{0})=\mathbf{0}$, and $\mathbf{h}:\Rn\rightarrow\Rm$ is continuously differentiable with $\mathbf{h}(\mathbf{0})=\mathbf{0}$.

\begin{definition}[\cite{ghallab2018extending}]\label{def:NNI}
The system \eqref{eq:nonlinear_system} is said to be \emph{nonlinear negative imaginary} (NNI) if there exists a continuously differentiable, positive-definite storage function $V:\Rn \rightarrow \R_+$ such that for all admissible inputs $\mathbf{u}(t)$ and corresponding solutions $\mathbf{x}(t)$, the output $\mathbf{y}(t) = \mathbf{h}(\mathbf{x}(t))$ is differentiable and
\begin{equation}\label{eq:NNI_condition}
    \dot{V}(\mathbf{x}(t)) \leq \dot{\mathbf{y}}^T(t) \mathbf{u}(t)
\end{equation}
holds for all $t \geq 0$ along system trajectories.
\end{definition}

\begin{remark}
The dissipative inequality \eqref{eq:NNI_condition} can be expressed in integral form by integrating from $0$ to $t$:
\begin{equation}\label{eq:NNI_integral}
  V(\mathbf{x}(t))\leq V(\mathbf{x}(0))+\int_{0}^{t}\dot{\mathbf{y}}^T(s)\mathbf{u}(s)\,ds
\end{equation}
for all $t \geq 0$.
\end{remark}

Following the analysis framework for passive systems \cite{isidori1999asymptotic}, we introduce stronger notions of the NNI property for feedback stability analysis.

\begin{definition}\label{def:MS_NNI}
The system \eqref{eq:nonlinear_system} is said to be \emph{marginally strictly nonlinear negative imaginary} (MS-NNI) if it is NNI and additionally, whenever $\mathbf{u}$, $\mathbf{x}$ satisfy
\begin{equation}\label{eq:equality_condition}
   \dot{V}(\mathbf{x}(t))=\dot{\mathbf{y}}^T(t)\mathbf{u}(t)
\end{equation}
for all $t>0$, then $\lim_{t\rightarrow\infty}\mathbf{u}(t)=\mathbf{0}$.
\end{definition}

\begin{definition}\label{def:WS_NNI}
The system \eqref{eq:nonlinear_system} is said to be \emph{weakly strictly nonlinear negative imaginary} (WS-NNI) if it is MS-NNI and globally asymptotically stable when $\mathbf{u}\equiv\mathbf{0}$.
\end{definition}

\begin{remark}\label{rem:SNI_WSNNI_equivalence}
For linear time-invariant systems, the SNI property (Definition~\ref{def:SNI}) is equivalent to the WS-NNI property (Definition~\ref{def:WS_NNI}); see \cite{11060877}.
\end{remark}

\subsection{Robust Stability of Positive Feedback Interconnections}

We now present a stability result for positive feedback interconnections of NNI systems, which forms the theoretical foundation for the control design in this paper. Consider two nonlinear systems $H_1$ and $H_2$ described by
\begin{subequations}\label{eq:H1_system}
\begin{align}
\dot{\mathbf{x}}_1 &= \mathbf{f}_1(\mathbf{x}_1,\mathbf{u}_1),\\
\mathbf{y}_1 &= \mathbf{h}_1(\mathbf{x}_1),
\end{align}
\end{subequations}
and
\begin{subequations}\label{eq:H2_system}
\begin{align}
\dot{\mathbf{x}}_2 &= \mathbf{f}_2(\mathbf{x}_2,\mathbf{u}_2),\\
\mathbf{y}_2 &= \mathbf{h}_2(\mathbf{x}_2),
\end{align}
\end{subequations}
respectively, with $\mathbf{f}_i:\Rn\times\Rm\rightarrow\Rn$ Lipschitz continuous satisfying $\mathbf{f}_i(\mathbf{0},\mathbf{0})=\mathbf{0}$, and $\mathbf{h}_i:\Rn\rightarrow\Rm$ continuously differentiable with $\mathbf{h}_i(\mathbf{0})=\mathbf{0}$ for $i=1,2$.

The following technical assumptions are required for the stability analysis.

Following \cite{11060877}, we have a set of technical assumptions on both systems $H_1$ and $H_2$, as well as on their open-loop interconnection (Fig.~\ref{sys13}).

\begin{assumption}\label{assum:steady_state_H1}
For any constant $\bar{\mathbf{u}}_1$, there exists a unique solution $(\bar{\mathbf{x}}_1, \bar{\mathbf{y}}_1)$ to
\begin{equation}\label{eq:ss_H1}
  \mathbf{0} = \mathbf{f}_1(\bar{\mathbf{x}}_1,\bar{\mathbf{u}}_1), \quad
  \bar{\mathbf{y}}_1= \mathbf{h}_1(\bar{\mathbf{x}}_1)
\end{equation}
such that $\bar{\mathbf{u}}_1\neq \mathbf{0}$ implies $\bar{\mathbf{x}}_1\neq \mathbf{0}$ and the mapping $\bar{\mathbf{u}}_1 \mapsto \bar{\mathbf{x}}_1$ is continuous.
\end{assumption}

\begin{assumption}\label{assum:steady_state_H2}
For any constant $\bar{\mathbf{u}}_2$, there exists a unique solution $(\bar{\mathbf{x}}_2, \bar{\mathbf{y}}_2)$ to
\begin{equation}\label{eq:ss_H2}
 \mathbf{0} = \mathbf{f}_2(\bar{\mathbf{x}}_2,\bar{\mathbf{u}}_2), \quad
 \bar{\mathbf{y}}_2= \mathbf{h}_2(\bar{\mathbf{x}}_2)
\end{equation}
such that $\bar{\mathbf{u}}_2\neq \mathbf{0}$ implies $\bar{\mathbf{x}}_2\neq \mathbf{0}$.
\end{assumption}

\begin{assumption}\label{assum:positivity}
For the steady-state solutions defined in Assumptions~\ref{assum:steady_state_H1} and \ref{assum:steady_state_H2}, we have
\begin{equation}
\mathbf{h}_1^T(\bar{\mathbf{x}}_{1})\mathbf{h}_2(\bar{\mathbf{x}}_{2})\geq 0
\end{equation}
for any constant $\bar{\mathbf{u}}_1$ with $\bar{\mathbf{u}}_2=\bar{\mathbf{y}}_1$.
\end{assumption}

\begin{assumption}\label{assum:sector_bound}
For the steady-state solutions defined in Assumptions~\ref{assum:steady_state_H1} and \ref{assum:steady_state_H2} with $\bar{\mathbf{u}}_2=\bar{\mathbf{y}}_1$, there exists a constant $\gamma\in(0,1)$ such that for any $\bar{\mathbf{u}}_1$, the following sector-bound condition is satisfied:
\begin{equation}\label{eq:sector_bound}
  \bar{\mathbf{y}}_2^T\bar{\mathbf{y}}_2\leq \gamma^2 \bar{\mathbf{u}}_1^T\bar{\mathbf{u}}_1.
\end{equation}
\end{assumption}

\begin{figure}[!t]
\centering
\tikzstyle{block} = [draw, thick, rectangle, rounded corners=5pt,
    minimum height=2em, minimum width=4em]
\tikzstyle{sum} = [draw, circle,inner sep=0pt,minimum size=1pt, node distance=1cm]
\tikzstyle{input} = [coordinate]
\tikzstyle{output} = [coordinate]
\tikzstyle{pinstyle} = [pin edge={to-,thin,black}]
\tikzstyle{int}=[draw, thick, rectangle, rounded corners=5pt,
    minimum height = 3em, minimum width = 3em]

\resizebox{\columnwidth}{!}{%
\begin{tikzpicture}[node distance=2.5cm,auto,>=latex']
\tikzstyle{block} = [draw, thick, rectangle, rounded corners=5pt,
    minimum height=3em, minimum width=5em]

    \node [int] (a) {$H_1$};
    \node (b) [left of=a,node distance=2cm, coordinate] {a};
    \node [int] (c) [right of=a, node distance=3cm] {$H_2$};
    \node [coordinate] (end) [right of=c, node distance=2cm]{};

    \path[->] (b) edge node {$\bar{u}_1$} (a);
    \path[->] (a) edge node {$\bar{y}_1  \quad \quad  \ \bar{u}_2$} (c);
    \draw[->] (c) edge node {$\bar{y}_2$} (end);

\end{tikzpicture}}%
\caption{Open-loop interconnection of systems $H_1$ and $H_2$ at steady state.}
\label{sys13}
\end{figure}
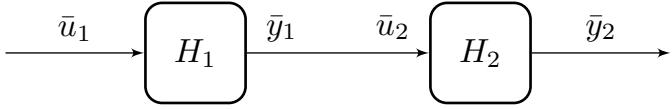

\tikzstyle{block} = [draw, thick, rectangle, rounded corners=5pt,
    minimum height=3em, minimum width=4em, align=center]
\tikzstyle{sum} = [draw, circle, node distance=1cm, inner sep=1pt, minimum size=8pt]
\tikzstyle{input} = [coordinate]
\tikzstyle{output} = [coordinate]
\tikzstyle{pinstyle} = [pin edge={to-,thin,black}]

The positive feedback interconnection shown in Figure~\ref{fig:feedback_interconnection} is described by the augmented closed-loop system
\begin{subequations}\label{eq:closed_loop}
\begin{align}
\dot{\mathbf{x}}_1 &= \mathbf{f}_1(\mathbf{x}_1, \mathbf{u}_1), \quad \mathbf{y}_1 = \mathbf{h}_1(\mathbf{x}_1), \\
\dot{\mathbf{x}}_2 &= \mathbf{f}_2(\mathbf{x}_2, \mathbf{y}_2), \quad \mathbf{y}_2 = \mathbf{h}_2(\mathbf{x}_2)
\end{align}
\end{subequations}
where $\mathbf{u}_1=\mathbf{y}_2$ and $\mathbf{y}_1=\mathbf{u}_2$.

The main stability result is stated in the following theorem.

\begin{theorem}[\cite{ghallab2018extending}]\label{thm:main_stability}
Consider the positive feedback interconnection of systems $H_1$ and $H_2$ shown in Figure~\ref{fig:feedback_interconnection}. Suppose $H_1$ is NNI and zero-state observable, $H_2$ is WS-NNI, and Assumptions~\ref{assum:steady_state_H1}--\ref{assum:sector_bound} are satisfied. Then the equilibrium point $(\mathbf{x}_1,\mathbf{x}_2)=\mathbf{0}$ of the closed-loop system \eqref{eq:closed_loop} is asymptotically stable.
\end{theorem}

\begin{remark}
A system is zero-state observable if, when the output is identically zero and the input is identically zero, the state must be identically zero. This property ensures that the system has no hidden unstable dynamics.
\end{remark}
\tikzstyle{block} = [draw, thick, rectangle, rounded corners=5pt,
    minimum height=3em, minimum width=4em, align=center]
\tikzstyle{sum} = [draw, circle, node distance=1cm, inner sep=1pt, minimum size=8pt]
\tikzstyle{input} = [coordinate]
\tikzstyle{output} = [coordinate]
\tikzstyle{pinstyle} = [pin edge={to-,thin,black}]
\begin{figure}[H]
\centering
\resizebox{7.5cm}{!}{%
\begin{tikzpicture}[auto, thick, node distance=2.5cm,>=latex']
    \node [input, name=input] {};
    \node [sum, right of=input] (sum1) {\small $+$};
    \node [block, right of=sum1, node distance=3cm] (system) {\large $H_1$};
    \node [coordinate, right of=system, node distance=2.8cm] (k) {};
    \node [block, below of=system, node distance=2.8cm] (measurements) {\large $H_2$};
    \node [coordinate, left of=measurements, node distance=2.8cm] (p) {};
    \node [sum, below of=k, node distance=2.8cm] (sum2) {\small $+$};
    \node [input, right of=sum2, node distance=1.1cm] (input2) {};

    \draw [->] (input) -- (sum1);
    \draw [->] (sum1) -- node[above] {\large $u_1$} (system);
    \draw [-] (system) -- node[above] {\large $y_1$} (k);
    \draw [->] (k) -| (sum2);
    \draw [->] (input2) -- node[above] {} (sum2);
    \draw [->] (sum2) -- node[above] {\large $u_2$} (measurements);
    \draw [-] (measurements) -- node[above] {\large $y_2$} (p);
    \draw [->] (p) -| (sum1);
\end{tikzpicture}}
\caption{Positive feedback interconnection of an NNI system $H_1$ and a WS-NNI system $H_2$.}
\label{fig:feedback_interconnection}
\end{figure}
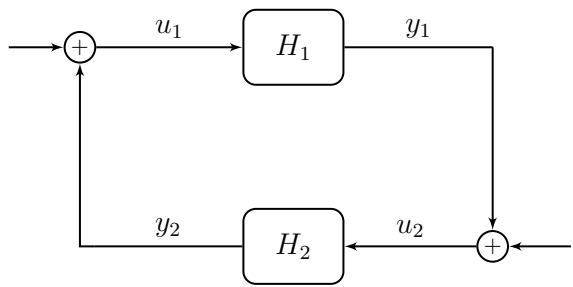

\section{Quadrotor System}\label{sec:quadrotor}

This section presents the kinematics and dynamics of the quadrotor system, establishing the foundation for the subsequent control design.

\subsection{Coordinate Frames and Kinematics}

Two reference frames are used to describe the quadrotor: an inertial frame fixed to the earth $\{R\}(O, x, y, z)$ and a body-fixed frame $\{R_B\}\{O_B, x_B, y_B, z_B\}$, where $O_B$ is fixed to the quadrotor's center of mass (see Figure~\ref{fig:quadrotor_diagram}). Frame $\{R_B\}$ is related to $\{R\}$ by a position vector $\bm{\xi}=[x, y, z]^T$ describing the center of gravity in $\{R_B\}$ relative to $\{R\}$, and by Euler angles $\bm{\eta} = [\phi, \theta, \psi]^T$ representing roll, pitch, and yaw. The Euler angles are bounded as:
\begin{equation*}
\phi \in(-\pi / 2, \pi / 2), \quad \theta \in(-\pi / 2, \pi / 2), \quad \psi \in(-\pi, \pi].
\end{equation*}

\begin{figure}[!t]
  \centering
  \includegraphics[width=0.45\textwidth]{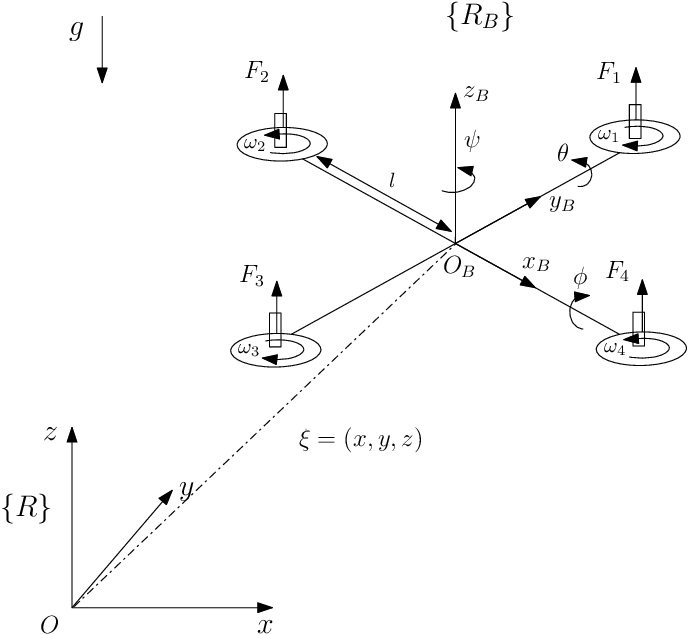}
  \caption{Quadrotor configuration showing body-fixed frame $\{B\}$ and inertial frame $\{E\}$.}
  \label{fig:quadrotor_diagram}
\end{figure}

The rotation matrix from body frame to inertial frame is
\begin{equation}\label{eq:rotation_matrix}
  \mathbf{R}_{B\to E}:=\left[
      \begin{array}{ccc}
        c_{\theta}c_{\psi} & c_{\psi}s_{\theta}s_{\phi}- c_{\phi}s_{\psi}& c_{\phi}c_{\psi}s_{\theta}+s_{\phi}s_{\psi} \\
         c_{\theta}s_{\psi} & s_{\theta}s_{\phi}s_{\psi}+ c_{\phi}c_{\psi} & c_{\phi}s_{\theta}s_{\psi}-c_{\psi}s_{\phi} \\
        -s_{\theta} & c_{\theta}s_{\phi} & c_{\theta}c_{\phi}
      \end{array}
    \right],
\end{equation}
where $s_{(\cdot)} := \sin(\cdot)$ and $c_{(\cdot)} := \cos(\cdot)$.

The angular velocity $\boldsymbol{\omega}=[p, q, r]^T$ in the body frame relates to the Euler angle rates through
\begin{equation}\label{eq:angular_velocity}
  \boldsymbol{\omega}=\mathbb{W}(\bm{\eta}) \dot{\bm{\eta}}, \quad
  \mathbb{W}(\bm{\eta}):=\left[
                \begin{array}{ccc}
                  1 & 0 & -s_{\theta} \\
                  0 & c_{\phi} & s_{\phi}c_{\theta} \\
                  0 & -s_{\phi} & c_{\phi}c_{\theta}
                \end{array}
              \right].
\end{equation}

\begin{remark}
Under small-angle approximations, $\mathbb{W}(\bm{\eta})\approx \mathbf{I}$, yielding $\boldsymbol{\omega}\approx\dot{\bm{\eta}}$.
\end{remark}

\subsection{Dynamics Model}

We employ the Euler--Lagrange approach to derive the equations of motion. The Lagrangian is
\begin{equation}\label{eq:lagrangian}
  \mathcal{L}(\mathbf{q},\dot{\mathbf{q}})=T_{\text{trans}}+T_{\text{rot}}-P,
\end{equation}
where $\mathbf{q}=[{\bm{\xi}}^T, \bm{\eta}^T]^T$ denotes the generalized coordinates, and
\begin{align}
T_{\text{trans}} &= \frac{1}{2} m \|\dot{\bm{\xi}}\|^2, \label{eq:T_trans}\\
T_{\text{rot}} &= \frac{1}{2} \dot{\bm{\eta}}^{T} \mathbf{J} \dot{\bm{\eta}}, \label{eq:T_rot}\\
P &= mgz \label{eq:potential}
\end{align}
are the translational kinetic energy, rotational kinetic energy, and gravitational potential energy, respectively. Here $m$ denotes the total mass, $g$ is gravitational acceleration, and
\begin{equation}
\mathbf{J}=\operatorname{diag}(J_{x}, J_{y}, J_{z})
\end{equation}
is the diagonal inertia matrix in the body frame.

The Euler--Lagrange equations
\begin{equation}\label{eq:euler_lagrange}
\frac{d}{dt}\left(\frac{\partial \mathcal{L}}{\partial \dot{\mathbf{q}}}\right)-\frac{\partial \mathcal{L}}{\partial \mathbf{q}}=\mathbf{F}
\end{equation}
with generalized forces $\mathbf{F}=[\mathbf{F}_{\bm{\xi}}^T, \bm{\tau}^T]^T$ yield the decoupled translational and rotational dynamics.

\begin{table}[!t]
\centering
\caption{Quadrotor Parameters}
\label{tab:quadrotor_parameters}
{\normalsize
\renewcommand{\arraystretch}{1.35}
\setlength{\tabcolsep}{10pt}
\begin{tabular}{lcc}
\toprule
\textbf{Parameter} & \textbf{Symbol} & \textbf{Unit}\\
\midrule
Quadrotor mass & $m$ & kg \\
Gravitational acceleration & $g$ & m/s$^{2}$ \\
Arm length & $l$ & m \\
Thrust coefficient & $b$ & N$\cdot$s$^{2}$ \\
Drag coefficient & $d$ & N$\cdot$m$\cdot$s$^{2}$ \\
Roll moment of inertia & $J_{x}$ & kg$\cdot$m$^{2}$ \\
Pitch moment of inertia & $J_{y}$ & kg$\cdot$m$^{2}$ \\
Yaw moment of inertia & $J_{z}$ & kg$\cdot$m$^{2}$ \\
Rotor inertia & $J_{r}$ & kg$\cdot$m$^{2}$ \\
\bottomrule
\end{tabular}}
\end{table}

The translational dynamics are
\begin{equation}\label{eq:translational_dynamics}
m \ddot{\bm{\xi}}+\begin{bmatrix}0\\0\\mg\end{bmatrix}=\mathbf{F}_{\bm{\xi}},
\end{equation}
where the thrust force in the inertial frame is
\begin{equation}\label{eq:thrust_force}
\mathbf{F}_{\bm{\xi}}
= u_1
\begin{bmatrix}
-\sin\theta \\
\cos\theta\sin\phi \\
\cos\theta\cos\phi
\end{bmatrix},
\end{equation}
with $u_1 = \sum_{i=1}^4 b \omega_i^2$ the total thrust magnitude generated by the four rotors spinning at angular velocities $\omega_i$.

The rotational dynamics are
\begin{equation}\label{eq:rotational_dynamics}
\mathbf{J} \ddot{\bm{\eta}}+\mathbf{C}(\bm{\eta}, \dot{\bm{\eta}}) \dot{\bm{\eta}}=\bm{\tau},
\end{equation}
where $\bm{\tau}=[\tau_\phi, \tau_\theta, \tau_\psi]^T$ is the control torque vector and $\mathbf{C}(\bm{\eta}, \dot{\bm{\eta}})$ captures Coriolis and centrifugal effects.

Introducing the complete state vector
\begin{equation}\label{eq:state_vector}
\mathbf{x}=[x, \dot{x}, y, \dot{y}, z, \dot{z}, \phi, \dot{\phi}, \theta, \dot{\theta}, \psi, \dot{\psi}]^{T}\in\R^{12}
\end{equation}
and control input vector $\mathbf{u}=[u_{1}, u_{2}, u_{3}, u_{4}]^{T}$, the complete quadrotor dynamics can be expressed in the standard form
\begin{equation}\label{eq:full_dynamics}
\dot{\mathbf{x}}=\mathbf{f}(\mathbf{x}, \mathbf{u}),
\end{equation}
where
\begin{equation}\label{full_dynamics}
\mathbf{f}(\mathbf{x}, \mathbf{u})=\left(\begin{array}{c}
x_{2} \\
-\frac{u_{1}}{m}\left(\sin x_{7} \sin x_{11}+\cos x_{7} \sin x_{9} \cos x_{11}\right) \\
x_{4} \\
\frac{u_{1}}{m}\left(\sin x_{7} \cos x_{11}-\cos x_{7} \sin x_{9} \sin x_{11}\right) \\
x_6\\
g-\frac{u_{1}}{m} \cos x_{7} \cos x_{9} \\[4pt]
\hdashline[4pt/4pt]
x_{8} \\
x_{10} x_{12} a_{1}-x_{10} a_{2} \omega(\bm{u})+b_{1} u_{2} \\
x_{10} \\
x_{8} x_{12} a_{3}+x_{8} a_{4} \omega(\bm{u})+b_{2} u_{3} \\
x_{12} \\
x_{10} x_{8} a_{5}+b_{3} u_{4}
\end{array}\right)
\end{equation}
and
\begin{equation*}
\begin{array}{l | l}
a_{1}=\left(J_{y}-J_{z}\right) / J_{x} \ & \ b_{1}=l / J_{x} \\
a_{2}=J_{r} / J_{x} \ & \ b_{2}=l / J_{y} \\
a_{3}=\left(J_{z}-J_{x}\right) / J_{y} \ & \ b_{3}=1 / J_{z} \\
a_{4}=J_{r} / J_{y} \\
a_{5}=\left(J_{x}-J_{y}\right) / J_{z} &
\end{array}
\end{equation*}

\section{Control Design and Stability Analysis}\label{sec:control}

This section develops the velocity-free position control strategy based on nonlinear negative imaginary systems theory. We focus on horizontal position control in the $(x,y)$ plane.

\subsection{Problem Formulation}

The control objective is to regulate the horizontal position $\bm{\xi}_h := [x, y]^T$ to a desired reference $\bm{\xi}_h^d$ using only position measurements. Without loss of generality, we set $\bm{\xi}_h^d = \mathbf{0}$ (regulation to the origin). The yaw angle is set to $\psi = 0$ as yaw rotation is not required for horizontal position control.

\subsection{Inner-Loop Controller: Proportional Feedback}

We first introduce proportional feedback to reshape the horizontal position dynamics. Define the control law
\begin{equation}\label{eq:inner_loop}
\mathbf{U} = -\mathbf{K}_p \bm{\xi}_h = -\mathbf{K}_p \begin{bmatrix} x \\ y \end{bmatrix},
\end{equation}
where $\mathbf{K}_p = \operatorname{diag}(k_p^x, k_p^y)$ with $k_p^x, k_p^y > 0$.

Let $\mathbf{F}_h := [F_x, F_y]^T$ denote the horizontal components of the thrust force. With the proportional feedback \eqref{eq:inner_loop}, the closed-loop horizontal position subsystem becomes
\begin{equation}\label{eq:horizontal_subsystem}
\left(\begin{array}{c}
\dot{x}_1 \\
\dot{x}_2 \\
\dot{x}_3 \\
\dot{x}_4
\end{array}\right)=\left(\begin{array}{c}
x_{2} \\
-\frac{u_{1}}{m}\cos x_{7} \sin x_{9}-k_{p}^x {x}_{1}\\
x_{4} \\
\frac{u_{1}}{m}\sin x_{7}-k_{p}^y{x}_{3}
\end{array}\right).
\end{equation}

Equivalently, in compact form
\begin{equation}\label{eq:compact_horizontal}
m \ddot{\bm{\xi}}_h + m\mathbf{K}_p \bm{\xi}_h = \mathbf{F}_h, \quad \mathbf{y} = \bm{\xi}_h.
\end{equation}

The following lemma establishes that this subsystem possesses the NNI property.

\begin{lemma}\label{lem:NNI_horizontal}
Consider the horizontal position subsystem \eqref{eq:horizontal_subsystem} with input $\mathbf{F}_h$ and output $\mathbf{y} = \bm{\xi}_h$. This system is nonlinear negative imaginary with respect to the storage function
\begin{equation}\label{eq:storage_function}
V(\bm{\xi}_h, \dot{\bm{\xi}}_h) = \frac{1}{2} m \|\dot{\bm{\xi}}_h\|^2 + \frac{1}{2}\bm{\xi}_h^T \mathbf{K}_p \bm{\xi}_h.
\end{equation}
\end{lemma}

\begin{proof}
The storage function \eqref{eq:storage_function} is positive definite and continuously differentiable. Computing its time derivative along trajectories of \eqref{eq:horizontal_subsystem}:
\begin{align}
\frac{dV}{dt} &= m\dot{\bm{\xi}}_h^T \ddot{\bm{\xi}}_h + \bm{\xi}_h^T \mathbf{K}_p \dot{\bm{\xi}}_h \nonumber\\
&= \dot{\bm{\xi}}_h^T \left(m\ddot{\bm{\xi}}_h + \mathbf{K}_p \bm{\xi}_h\right) \nonumber\\
&= \dot{\bm{\xi}}_h^T \mathbf{F}_h, \label{eq:dissipation_proof}
\end{align}
where we have used \eqref{eq:compact_horizontal} in the last equality. Since $\mathbf{y} = \bm{\xi}_h$, we have $\dot{\mathbf{y}} = \dot{\bm{\xi}}_h$, and thus
\begin{equation}
\dot{V} = \dot{\mathbf{y}}^T \mathbf{F}_h,
\end{equation}
which establishes the NNI property according to Definition~\ref{def:NNI}.
\end{proof}

\subsection{Outer-Loop Controller: Strictly Negative Imaginary Design}

For the outer loop, we employ a strictly negative imaginary integral resonant controller with transfer function matrix
\begin{equation}\label{eq:SNI_controller}
\mathbf{C}(s) = [s\mathbf{I} + \bm{\Gamma}\bm{\Delta}]^{-1}\bm{\Gamma},
\end{equation}
where $\bm{\Gamma} = \operatorname{diag}(\Gamma, \Gamma)$ and $\bm{\Delta} = \operatorname{diag}(\delta, \delta)$ are positive-definite diagonal matrices with tuning parameters $\Gamma > 0$ and $\delta > 0$. The transfer function matrix $ \mathbf{C}_{\bm{v}}(s)$ is strictly negative imaginary (see e.g. \cite{petersen2010}). The dc-gain (the gain of the system at steady-state) of the controller is $ \mathbf{C}_{\bm{v}}(0)=\bm{\Delta}^{-1}$.

A state-space realization of \eqref{eq:SNI_controller} is
\begin{subequations}\label{eq:controller_state_space}
\begin{align}
\dot{\mathbf{z}}_c &= -\bm{\Gamma}\bm{\Delta}\mathbf{z}_c + \bm{\Gamma}\bm{\xi}_h, \\
\mathbf{u}_c &= \mathbf{z}_c,
\end{align}
\end{subequations}
where $\mathbf{z}_c := [z_1, z_2]^T$ is the controller state and $\mathbf{u}_c$ is the controller output.

Explicitly, the controller dynamics are
\begin{align}
\dot{z}_1 &= -\Gamma\delta z_1 + \Gamma x, \label{eq:controller_x}\\
\dot{z}_2 &= -\Gamma\delta z_2 + \Gamma y, \label{eq:controller_y}
\end{align}
with control inputs $F_x = z_1$ and $F_y = z_2$.

To apply Theorem~\ref{thm:main_stability}, we must verify Assumptions~\ref{assum:steady_state_H1}--\ref{assum:sector_bound} for the open-loop interconnection of the horizontal position subsystem and the SNI controller.

At steady state with constant input $\mathbf{\bar{F}}_h$, the dynamics \eqref{eq:horizontal_subsystem} yield
\begin{align}
0 &= \frac{\bar{F}_x}{m} - k_p^x \bar{x}, \label{eq:ss_x}\\
0 &= \frac{\bar{F}_y}{m} - k_p^y \bar{y}. \label{eq:ss_y}
\end{align}

Solving for the steady-state positions, we obtain
\begin{equation}\label{eq:steady_state_solution}
\bar{x} = \frac{\bar{F}_x}{mk_p^x}, \quad \bar{y} = \frac{\bar{F}_y}{mk_p^y}.
\end{equation}

These equations define a unique, continuous mapping $\bar{\mathbf{F}}_h \mapsto \bar{\bm{\xi}}_h$, satisfying Assumption~\ref{assum:steady_state_H1}.

\begin{figure}[!t]
\centering
\tikzstyle{block} = [draw, thick, rectangle, rounded corners=5pt,
    minimum height=2em, minimum width=4em]
\tikzstyle{sum} = [draw, circle,inner sep=0pt,minimum size=1pt, node distance=1cm]
\tikzstyle{input} = [coordinate]
\tikzstyle{output} = [coordinate]
\tikzstyle{pinstyle} = [pin edge={to-,thin,black}]
\tikzstyle{int}=[draw, thick, rectangle, rounded corners=5pt,
    minimum height = 3em, minimum width = 3em]

\resizebox{\columnwidth}{!}{%
\begin{tikzpicture}[node distance=2.5cm,auto,>=latex']
\tikzstyle{block} = [draw, thick, rectangle, rounded corners=5pt,
    minimum height=3em, minimum width=5em]

    \node [int] (a) {\small\shortstack{Reshaped Position\\ Subsystem \eqref{eq:horizontal_subsystem}}};
    \node (b) [left of=a,node distance=2.5cm, coordinate] {a};
    \node [int] (c) [right of=a, node distance=4.2cm] {\small\shortstack{SNI\\ Controller \eqref{eq:SNI_controller}}};
    \node [coordinate] (end) [right of=c, node distance=2cm]{};

    \path[->] (b) edge node {$\bar{\mathbf{F}}_h$} (a);
    \path[->] (a) edge node {$\bar{\bm{\xi}}_h $} (c);
    \draw[->] (c) edge node {$\bar{\mathbf{u}}_c$} (end);
\end{tikzpicture}}%
\caption{Open-loop interconnection of reshaped position subsystem \eqref{eq:horizontal_subsystem} and controller \eqref{eq:SNI_controller} at steady state.}
\label{sys14}
\end{figure}
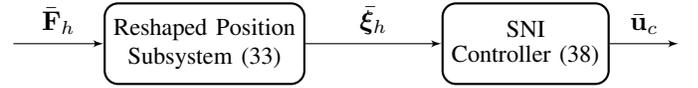

The SNI controller \eqref{eq:controller_state_space} is linear with well-defined steady-state gain $\bm{\Delta}^{-1}$, trivially satisfying Assumption~\ref{assum:steady_state_H2}.

Also, at steady state, the controller output is $\bar{\mathbf{u}}_c = \bm{\Delta}^{-1}\bar{\bm{\xi}}_h$. Thus,
\begin{equation}
\bar{\bm{\xi}}_h^{\,T} \bar{\mathbf{u}}_c = \bar{\bm{\xi}}_h^{\,T} \bm{\Delta}^{-1}\bar{\bm{\xi}}_h = \frac{1}{\delta}\|\bar{\bm{\xi}}_h\|^2 \geq 0,
\end{equation}
satisfying Assumption~\ref{assum:positivity}.

The sector-bound condition \eqref{eq:sector_bound} requires
\begin{equation}\label{eq:sector_bound_specific}
\|\bar{\mathbf{u}}_c\|^2 \leq \gamma^2 \|\bar{\mathbf{F}}_h\|^2
\end{equation}
for some $\gamma \in (0,1)$.

From the steady-state relations
\begin{equation}\label{eq_47}
\|\bar{\mathbf{u}}_c\|^2 = \frac{1}{\delta^2}\|\bar{\bm{\xi}}_h\|^2 = \frac{1}{\delta^2}\left(\bar{x}^2 + \bar{y}^2\right).
\end{equation}

Substituting \eqref{eq:steady_state_solution} into \eqref{eq_47}, we obtain
\begin{equation}
\|\bar{\mathbf{u}}_c\|^2 = \frac{1}{\delta^2}\left[\frac{\bar{F}_x^2}{(mk_p^x)^2} + \frac{\bar{F}_y^2}{(mk_p^y)^2}\right].
\end{equation}

For \eqref{eq:sector_bound_specific} to hold for all $\bar{\mathbf{F}}_h$, it is sufficient that
\begin{equation}\label{eq:sufficient_condition}
\frac{1}{\delta^2} \max\left\{\frac{1}{(mk_p^x)^2}, \frac{1}{(mk_p^y)^2}\right\} \leq \gamma^2.
\end{equation}

This yields the explicit controller parameter constraint
\begin{equation}\label{eq:parameter_constraint}
\delta^2 \geq \frac{1}{\gamma^2 m^2 \min\{(k_p^x)^2, (k_p^y)^2\}}\cdot
\end{equation}

\begin{remark}
Condition \eqref{eq:parameter_constraint} provides the explicit design guidance that the controller parameter $\delta$ must be chosen sufficiently large relative to the minimum inner-loop proportional gain to ensure the sector-bound condition holds.
\end{remark}

The next lemma is needed to establish the main stability result.

\begin{lemma}\label{lem:zero_state_observability}
The horizontal position subsystem \eqref{eq:horizontal_subsystem} is zero-state observable.
\end{lemma}

\begin{proof}
Consider the case where the output is zero ($\mathbf{y} = \bm{\xi}_h = \mathbf{0}$) and the input is zero ($\mathbf{F}_h = \mathbf{0}$). This implies $x_1 = 0$ and $x_3 = 0$. From \eqref{eq:horizontal_subsystem}, the acceleration equations become
\begin{align}
\dot{x}_2 &= -k_p^x x_1 = 0, \\
\dot{x}_4 &= -k_p^y x_3 = 0.
\end{align}
Since $k_p^x, k_p^y > 0$ and $x_1 = x_3 = 0$, we have $\dot{x}_2 = \dot{x}_4 = 0$. The position dynamics $\dot{x}_1 = x_2$ and $\dot{x}_3 = x_4$ with $x_1 = x_3 = 0$ imply $x_2 = x_4 = 0$. Thus, the entire state vector is zero, establishing zero-state observability.
\end{proof}

\subsection{Main Stability Result}

We now state the main stability theorem for the proposed control scheme. The overall closed-loop architecture is illustrated in Figure~\ref{fig:control_diagram}, which depicts the inner proportional loop and the outer SNI controller acting in positive feedback.

\tikzset{
    boldsignal/.style={line width=1pt, ->, >=Latex},
    block/.style={
        draw,
        rounded corners=3pt,
        minimum height=2.8em,
        minimum width=8.5em,
        align=center,
        line width=1pt
    },
    sum/.style={
        draw,
        circle,
        minimum size=6.5mm,
        line width=1pt
    },
    gain/.style={
        draw,
        rounded corners=3pt,
        minimum height=2.4em,
        minimum width=5em,
        align=center,
        line width=1pt
    },
    controller/.style={
        draw,
        rounded corners=3pt,
        minimum height=3.2em,
        minimum width=8em,
        align=center,
        line width=1pt
    }
}

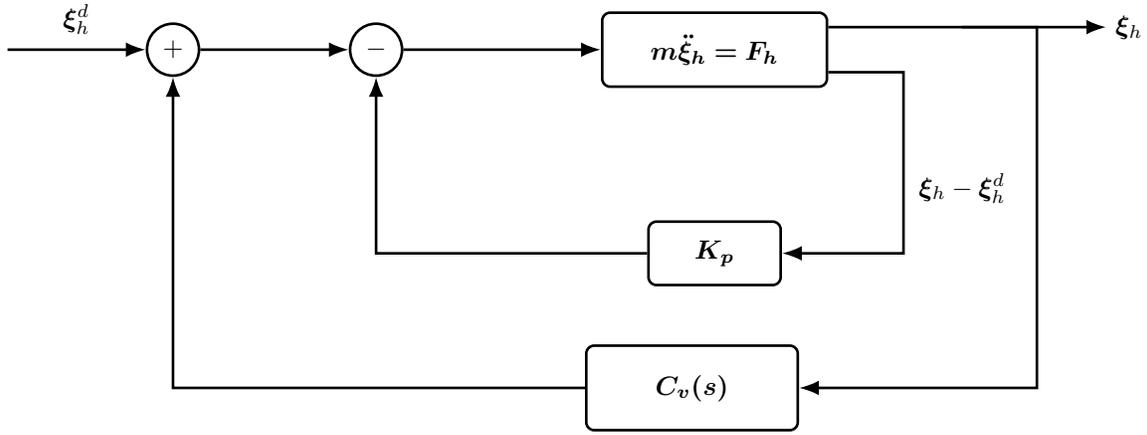
\begin{figure*}\centering
\begin{tikzpicture}[node distance=2.4cm]

\node[sum] (sum1) at (0,0) {$+$};
\node[sum] (sum2) at (2.7,0) {$-$};

\node[block] (plant) at (7.2,0)
{$\boldsymbol{m\ddot{\xi}_h = F_{h}}$};

\node (xilabel) at ($(plant.east)+(4.0cm,0.3)$)
{$\boldsymbol{\xi}_h$};

\node[gain] (kp) at ($(plant.south)+(0,-2.2)$)
{$\boldsymbol{K_p}$};

\node[controller] (sni) at (6.9,-4.5)
{$\boldsymbol{C_v(s)}$};


\draw[boldsignal] ($(sum1)+(-2.2,0)$) -- (sum1)
 node[midway, above=2pt] {$\boldsymbol{\xi}_h^d$};

\draw[boldsignal] (sum1) -- (sum2);
\draw[boldsignal] (sum2) -- (plant);


\draw[boldsignal] ($(plant.east)+(0,0.3)$) --
                  (xilabel) coordinate[pos=0.48](outertap);

\draw[boldsignal] ($(plant.east)+(0,-0.3)$) -- ++(1.0,0)
    |- (kp.east)
    node[pos=0.32, right=2pt] {$\boldsymbol{\xi}_h - \boldsymbol{\xi}_h^d$};

\draw[boldsignal] (kp.west) -| (sum2.south);


\coordinate (outertapshift) at ($(outertap)+(1.0cm,0)$);

\draw[boldsignal] (outertap) --
                  (outertapshift) |- (sni.east);

\draw[boldsignal] (sni.west) -| (sum1.south);

\end{tikzpicture}
\caption{Block diagram of the proposed velocity-free horizontal position control system.
         The inner loop feeds back position $\bm{\xi}_h$ through the proportional gain
         $\boldsymbol{K_p}$; the outer loop feeds back through
         the SNI controller $C_v(s)$.}
\label{fig:control_diagram}
\end{figure*}
\begin{theorem}\label{thm:main_result}
Consider the closed-loop system shown in Figure~\ref{fig:control_diagram}, consisting of the quadrotor horizontal position subsystem \eqref{eq:horizontal_subsystem} and the SNI controller \eqref{eq:controller_state_space} in positive feedback configuration. Suppose the controller parameters satisfy condition \eqref{eq:parameter_constraint} with $\gamma \in (0,1)$. Then the origin of the closed-loop system is asymptotically stable, and the horizontal position converges to zero, \ie \ $\lim_{t\to\infty}\bm{\xi}_h(t) = \mathbf{0}$.
\end{theorem}

\begin{proof}
The result follows from Theorem~\ref{thm:main_stability} by noting that the horizontal position subsystem is NNI (Lemma~\ref{lem:NNI_horizontal}) and zero-state observable (Lemma~\ref{lem:zero_state_observability}), the SNI controller is WS-NNI (Remark~\ref{rem:SNI_WSNNI_equivalence}), and Assumptions~\ref{assum:steady_state_H1}--\ref{assum:sector_bound} are satisfied as verified above. Therefore, asymptotic stability of the equilibrium point $(\bm{\xi}_h, \mathbf{z}_c) = \mathbf{0}$ is guaranteed.
\end{proof}

\begin{remark}
Theorem~\ref{thm:main_result} guarantees robust stability in the sense that stability is maintained for all parameter values satisfying \eqref{eq:parameter_constraint}, regardless of model uncertainties in the quadrotor dynamics, provided the NNI property is preserved.
\end{remark}

\section{Simulation Results}\label{sec:simulation}

Numerical simulations are presented to validate the theoretical results and demonstrate the effectiveness of the proposed control approach. The quadrotor parameters used in the simulation are listed in Table~\ref{tab:simulation_parameters}. These values represent a typical small-scale quadrotor platform.

\begin{table}[!t]
\centering
\caption{Simulation Parameters}
\label{tab:simulation_parameters}
{\normalsize
\renewcommand{\arraystretch}{1.35}
\setlength{\tabcolsep}{10pt}
\begin{tabular}{lc}
\toprule
\textbf{Parameter} & \textbf{Value} \\
\midrule
Mass $m$ & 0.5 kg \\
Moments of inertia $J_{xx} = J_{yy}$ & $4.85 \times 10^{-3}$ kg$\cdot$m$^2$ \\
Moment of inertia $J_{zz}$ & $8.81 \times 10^{-3}$ kg$\cdot$m$^2$ \\
Gravitational acceleration $g$ & 9.81 m/s$^2$ \\
Thrust coefficient $b$ & $2.92 \times 10^{-6}$ N$\cdot$s$^2$ \\
Drag coefficient $d$ & $1.12 \times 10^{-7}$ N$\cdot$m$\cdot$s$^2$ \\
\midrule
Inner-loop gains $k_p^x = k_p^y$ & 5 \\
Outer-loop parameters $\bm{\Delta}$ & $\operatorname{diag}(0.6, 0.6)$ \\
Outer-loop parameters $\bm{\Gamma}$ & $\operatorname{diag}(160, 160)$ \\
\bottomrule
\end{tabular}}
\end{table}

The controller parameters are selected according to condition \eqref{eq:parameter_constraint}. With $m = 0.5$ kg, $k_p^x = k_p^y = 5$, and choosing $\gamma = 0.8$, the minimum required value is
\begin{equation}
\delta_{\min}^2 = \frac{1}{0.64 \times 0.25 \times 25} = 0.25 \quad \Rightarrow \quad \delta_{\min} = 0.5
\end{equation}

The selected value $\delta = 0.6$ satisfies this requirement.

The simulation scenario considers regulation of the horizontal position to the origin from the initial condition
\begin{equation}
\bm{\xi}_h(0) = \begin{bmatrix} 2 \\ -1.5 \end{bmatrix} \text{ m}, \quad
\dot{\bm{\xi}}_h(0) = \begin{bmatrix} 0 \\ 0 \end{bmatrix} \text{ m/s}.
\end{equation}

The control objective is to drive the horizontal position to the desired reference $\bm{\xi}_h^d = \mathbf{0}$.

Figure~\ref{fig:position_tracking} shows the time evolution of the horizontal position components. Both $x$ and $y$ positions converge smoothly and asymptotically to zero, confirming the stability result of Theorem~\ref{thm:main_result}. The convergence is achieved without velocity measurements, relying solely on position feedback through the NI controller design.

\begin{figure*}[!t]
  \centering
  \includegraphics[width=2\columnwidth]{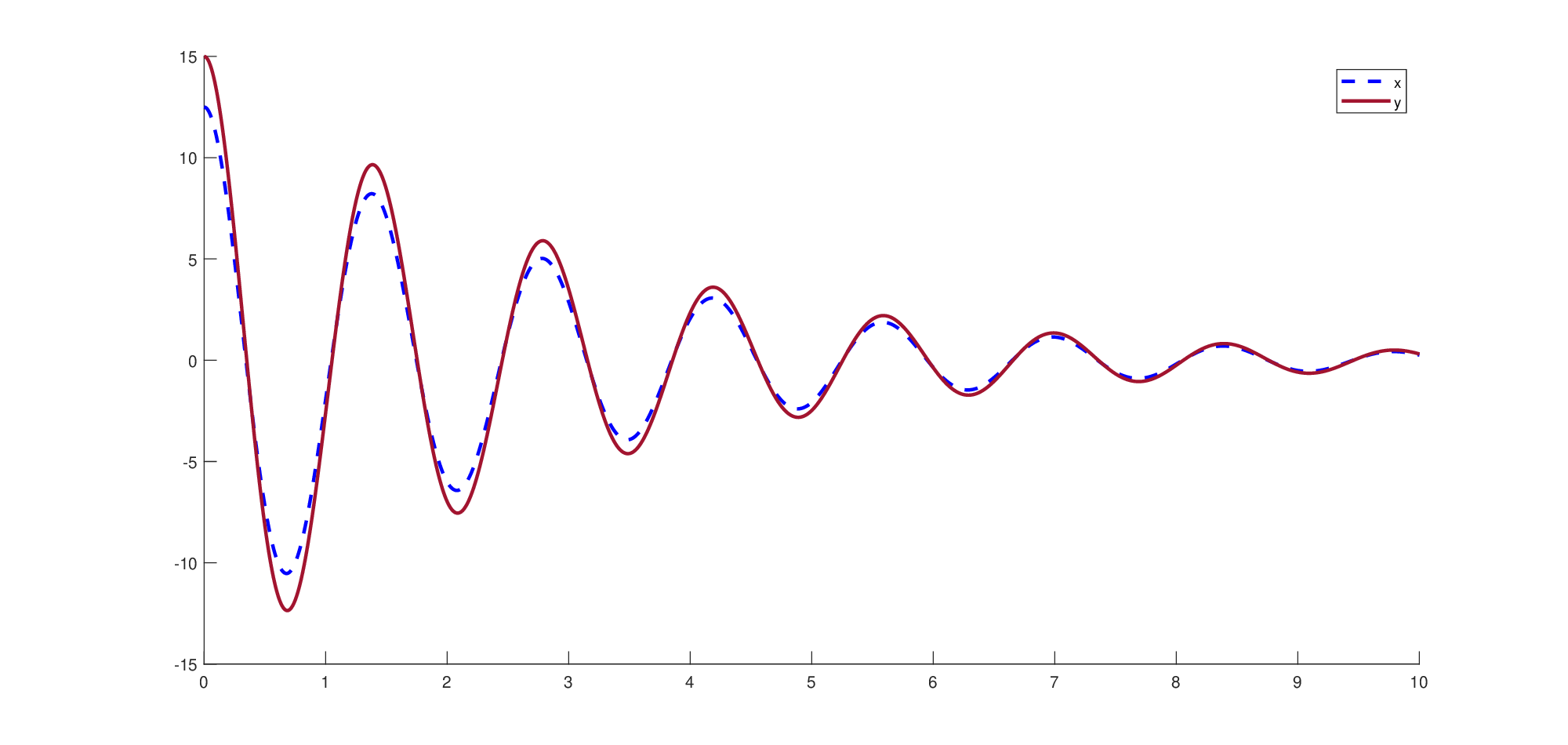}
  \caption{Time evolution of horizontal position components $x(t)$ and $y(t)$, demonstrating asymptotic convergence to the origin.}
  \label{fig:position_tracking}
\end{figure*}

The simulation results demonstrate several key features of the proposed control scheme. First, velocity-free operation is achieved, with asymptotic stability attained without requiring velocity measurements or observers. Second, the position trajectories exhibit smooth transient response with no oscillations or overshoot. Third, the controller demonstrates robustness, as the parameters satisfy the sector-bound condition with margin, ensuring robust stability.

\section{Conclusions}\label{sec:conclusions}

This paper has presented a velocity-free horizontal position control approach for quadrotor systems using nonlinear negative imaginary systems theory.
The first key contribution is the establishment of a theoretical framework showing that the quadrotor horizontal position subsystem with proportional feedback exhibits the nonlinear negative imaginary property, enabling application of NI stability theory. Building on this foundation, a strictly negative imaginary integral resonant controller was designed for the outer loop, achieving asymptotic stability through a positive feedback architecture.
We further derived explicit sector-bound conditions relating controller and plant parameters, providing clear design guidelines for practical implementation. Most significantly, the proposed approach achieves robust asymptotic stability using only position measurements, eliminating the need for velocity sensors or observers, which represents a substantial practical advantage for quadrotor control systems.

\bibliographystyle{IEEEtran}
\bibliography{ifacbib2}

\end{document}